\documentclass[10pt,a4paper]{article}
\usepackage[utf8]{inputenc}
\usepackage[T1]{fontenc}

\usepackage{amsmath}
\usepackage{amsfonts}
\usepackage{amssymb}
\usepackage{framed}
\usepackage{mathtools}
\usepackage{xspace}
\usepackage[amsthm,thmmarks]{ntheorem}
\usepackage{fullpage}

\usepackage{xspace}
\usepackage{graphicx}
\usepackage{cite}
\usepackage{colortbl}
\usepackage{xcolor}
\usepackage{placeins}

\usepackage{tikz}
\usetikzlibrary{patterns}
\usepackage[small]{complexity}
\newclass{\PromiseNP}{PromiseNP}

\makeatletter
\g@addto@macro\@floatboxreset\centering
\makeatother

\newcount\bsubfloatcount
\newtoks\bsubfloattoks
\newdimen\bsubfloatht

\makeatletter
\newcommand{\bsubfloat}[2][]{%
  \sbox\z@{#2}%
  \ifdim\bsubfloatht<\ht\z@
    \bsubfloatht=\ht\z@
  \fi
  \advance\bsubfloatcount\@ne
  \@namedef{bsubfloat\romannumeral\bsubfloatcount}{%
    \subfloat[#1]{\vbox to\bsubfloatht{\hbox{#2}\vfill}}}%
}
\newcommand{\resetbsubfloat}{\bsubfloatcount\z@\bsubfloatht=\z@}
\makeatother

\newcommand{\ProblemFormat}[1]{{\sc #1}}
\newcommand{\ProblemName}[1]{\ProblemFormat{#1}\xspace}

\newcommand{\probThreeSAT}[0]{\ProblemName{$3$-SAT}}
\newcommand{\probFourSAT}[0]{\ProblemName{$4$-SAT}}
\newcommand{\probThreeCol}[0]{\ProblemName{$3$-Coloring}}
\newcommand{\probPlanarThreeSAT}[0]{\ProblemName{Planar $3$-SAT}}
\newcommand{\probPlanarOneInThreeSAT}[0]{\ProblemName{Planar $1$-in-$3$-SAT}}

\newcommand{\probrFix}{\ProblemName{$r$-Fix}}
\newcommand{\probrFixPromise}{\ProblemName{$r$-Fix-Promise}}
\newcommand{\probrFixSwap}{\ProblemName{$r$-Swap}}
\newcommand{\probrFixSwapPromise}{\ProblemName{$r$-Swap-Promise}}

\newcommand{\probFix}{\ProblemName{Fix}}
\newcommand{\probFixSwap}{\ProblemName{Swap}}

\newcommand{\probPlanarrFix}{\ProblemName{Planar $r$-Fix}}

\newcommand{\probThreeCV}[0]{\ProblemName{$3$-Swap-Promise}}
\newcommand{\probThreeSwapPromise}[0]{\ProblemName{$3$-Fix-Promise}}
\newcommand{\probPlanarThreeCV}{\ProblemName{Planar $3$-Swap-Promise}}
\newcommand{\probPrExt}{\ProblemName{Precoloring Extension}}
\newcommand{\probkPrExt}{\ProblemName{$r$-Precoloring Extension}}
\newcommand{\probPrExtShort}{\ProblemName{PrExt}}
\newcommand{\probkPrExtShort}{\ProblemName{$r$-PrExt}}
\newcommand{\probIndSet}{\ProblemName{Independent Set}}
\newcommand{\probClique}{\ProblemName{Clique}}

\newcommand{\heading}[1]{\medskip\noindent{\bf #1.\ }}%

\newtheorem{theorem}{Theorem}
\newtheorem{lemma}[theorem]{Lemma}
\newtheorem{corollary}[theorem]{Corollary}
\newtheorem{conjecture}[theorem]{Conjecture}
\newtheorem{definition}[theorem]{Definition}

\title{On the Complexity of Restoring Corrupted Colorings\thanks{Work partially supported by the Emil Aaltonen Foundation (J.L).}}

\author{Marzio De Biasi\thanks{E-mail: \texttt{marziodebiasi@gmail.com}} \and Juho Lauri\thanks{Tampere University of Technology, Finland. E-mail: \texttt{juho.lauri@tut.fi}}}

\begin{document}

\maketitle

\begin{abstract}
In the \probrFix problem, we are given a graph $G$, a (non-proper) vertex-coloring $c : V(G) \to [r]$, and a positive integer~$k$. The goal is to decide whether a proper $r$-coloring $c'$ is obtainable from~$c$ by recoloring at most $k$ vertices of $G$. Recently, Junosza-Szaniawski, Liedloff, and Rz{\k{a}}{\.z}ewski [SOFSEM 2015] asked whether the problem has a polynomial kernel parameterized by the number of recolorings~$k$. In a full version of the manuscript, the authors together with Garnero and Montealegre, answered the question in the negative: for every $r \geq 3$, the problem \probrFix does not admit a polynomial kernel unless $\NP \subseteq \coNP / \poly$. Independently of their work, we give an alternative proof of the theorem.
Furthermore, we study the complexity of \probrFixSwap, where the only difference from \probrFix is that instead of~$k$ recolorings we have a budget of~$k$ color swaps. We show that for every $r \geq 3$, the problem \probrFixSwap is $\W[1]$-hard whereas \probrFix is known to be FPT. Moreover, when~$r$ is part of the input, we observe both \probFix and \probFixSwap are $\W[1]$-hard parameterized by treewidth. We also study promise variants of the problems, where we are guaranteed that a proper $r$-coloring~$c'$ is indeed obtainable from~$c$ by some finite number of swaps. For instance, we prove that for $r=3$, the problems \probrFixPromise and \probrFixSwapPromise are $\NP$-hard for planar graphs. As a consequence of our reduction, the problems cannot be solved in $2^{o(\sqrt{n})}$ time unless the Exponential Time Hypothesis (ETH) fails.
\end{abstract}

\section{Introduction}
Computational models are sometimes too idealized, and do not capture all information available or relevant to a problem. Moreover, in a dynamic world, constraints change over time due to more information becoming available. A problem arising frequently in practice in e.g., scheduling~\cite{Marx2004} is graph coloring. An assignment of $r$ colors to the vertices of a graph $G=(V,E)$ is an \emph{$r$-coloring}. Formally, an $r$-coloring $c : V \to [r]$ is said to be \emph{proper} if $c(u) \neq c(v)$ for every $uv \in E$. In the graph coloring problem, the goal is to find the smallest $r$ for which a graph is $r$-colorable. This quantity is known as the \emph{chromatic number} of $G$, and is denoted by $\chi(G)$. Graph coloring is one of the most central problems in discrete mathematics and optimization. For a general introduction to the topic, we refer the reader to the book~\cite{Jensen2011}. 

Suppose we have a graph $G$ for which we have computed a proper vertex-coloring using $\chi(G)$ colors. Then due to constraints changing, a malicious adversary, or e.g., a system failure, the coloring is corrupted by redistributing the vertex colors over the graph $G$. How hard is it to restore the corrupted coloring to some optimal proper coloring of $G$? In 2006, Clark, Holliday, Holliday, Johnson, Trimm, Rubalcaba, and Walsh~\cite{Clark2006} introduced and investigated this problem under the name \emph{chromatic villainy}. For restoring a corrupted coloring, they used the following operation called a \emph{swap}. A swap between two distinct vertices $u$ and $v$ is the operation of assigning to $u$ the color that appears on $v$, and vice versa. Formally, let $c$ be a vertex-coloring of a graph $G$. The \emph{villainy} of $c$, denoted by $B(c)$, is the minimum number of swaps to be performed to transform $c$ into some proper vertex-coloring of $G$. The quantity $B(c)$ can also be seen as the minimum number of recolorings with the constraint that each color used in the new coloring $c'$ must be used the same number of times as in $c$. %For the worst case over all optimal proper colorings, let $B(G) = \max_c B(c)$. 
%Among other results, it is shown in~\cite{Clark2006} that complete graphs are precisely the graphs with villainy equal to zero (excluding the empty graph). 
In addition to the graph-theoretic viewpoint, there has been interest in the computational aspects of chromatic villainy~\cite{Regs2013}. 

Very recently, Junosza-Szaniawski, Liedloff, and Rz{\k{a}}{\.z}ewski~\cite{Junosza2015} studied a problem they call \probrFix. In the problem, we are given a graph $G=(V,E)$, a non-proper $r$-coloring $c$ of $V(G)$, and an integer $k$. The task is to decide whether a proper vertex-coloring $c'$ is obtainable from $c$ using at most $k$ vertex recolorings. The authors observe the problem is $\NP$-complete, and give several positive complexity results as well. In particular, using the framework of Bj{\"o}rklund, Husfeldt, and Koivisto~\cite{Bjorklund2009} they obtain a $O^*(2^n)$-time and exponential space algorithm, where $n$ is the number of vertices in the input graph. Furthermore, they show that for any fixed $r$, the problem is fixed-parameter tractable (FPT) parameterized by the number of recolorings $k$. Finally, in the same paper, the authors show that for graphs of treewidth $t$, the problem can be solved in $O(nr^{t+2})$ time. For a discussion on several related reoptimization and reconfiguration problems, we refer the reader to~\cite{Junosza2015}. We also note that their results are considerably expanded in a full version of the manuscript currently under review~\cite{fix-journal}.

The main difference between \probrFix and chromatic villainy as defined by Clark~\emph{et al.}~\cite{Clark2006} is the basic operation used. In \probrFix it is a recoloring, whereas in chromatic villainy it is a swap. For this reason, we shall refer to the computational problem arising from chromatic villainy as \probrFixSwap. That is, the input to \probrFixSwap is exactly the same as it is in \probrFix, but instead of $k$ recolorings we have a budget of $k$ swaps. Strictly speaking, there is another difference. In fact, in chromatic villainy, the corrupted vertex-coloring is not an arbitrary one, but has an additional property that could possibly be exploited. Namely, the property is that by using some finite number swaps, a proper vertex-coloring is indeed obtainable from the given one. Additional promises of properties of a problem are captured by the notion of a \emph{promise problem}. Promise problems were introduced and studied by Even, Selman, and Yacobi~\cite{Even1984}, and they have several applications (for a survey, see~\cite{Goldreich2006}). In fact, it seems fair to argue that promise problems model real-world problems more accurately. Indeed, Goldreich~\cite{Goldreich2006} writes: ``... I contend that almost all readers refer to this notion when thinking about computational problems, although they may be often unaware of this fact''. Motivated by these facts, we also separate two additional problems, \probrFixPromise and \probrFixSwapPromise, for which the input is guaranteed to satisfy additional properties (a precise definition is provided in Section~\ref{sec:prelims}).

In this context, it is natural to investigate whether a hard problem becomes easy when the set of instances is restricted. A priori, it is unknown how the addition of a promise affects the computational complexity of a problem. For example, it is hard to decide whether a graph contains a Hamiltonian cycle, even if we are promised it contains at most one such cycle~\cite{Johnson1985}. In the other direction, one can efficiently find a satisfying assignment for an $n$-variable SAT formula that promises to have at least $2^n/n$ satisfying assignments. However, as shown by Valiant and Vazirani~\cite{Valiant1985}, it is hard to find an assignment even when the formula is promised to have exactly one solution.

\heading{Motivation}
The problems \probrFix and \probrFixSwap are tightly related to local search which is a core technique in solving combinatorial optimization and operations research problems in practice. In local search, one aims to improve upon the current solution by replacing it with a better solution in its neighborhood. Specifically, the neighborhood is defined by the set of allowed operations that modify the current solution. Plausibly, the larger the neighborhood, the less likely is the local search to get stuck in a local optimum. On the other hand, the allowed operations should not be too demanding to compute. In fact, there has been significant interest in applying methods from parameterized complexity to analyzing local search procedures (see e.g.,~\cite{Khuller2003,Krokhin2012,Fellows2012,Szeider2011,Berg2016}). 

The studied problems are also related to combinatorial reconfiguration problems, in particular \emph{$r$-coloring reconfiguration}~\cite{Johnson2014,Bonsma2014,Wrochna2014}. In this problem, we are given two proper $r$-colorings for a graph, and asked whether one can be transformed into the other by changing one color at a time, maintaining a proper coloring throughout. We argue that in the context of local search, the property of maintaining a proper $r$-coloring at each step can be relaxed: we are only interested in eventually arriving at a solution. To further motivate the use of promise conditions, we remark that there are $r$-coloring problems in which we know the sizes of the color classes (if an $r$-coloring exists). These include well-known problems arising from coding theory, such as partitioning of the $n$-dimensional Hamming space into binary codes with certain properties~\cite{Ostergard2004}.

\heading{Our results}
We continue the investigation of the complexity of restoring corrupted colorings. Specifically, we further study the complexity of \probrFix, under different basic operations and/or promise conditions.
\begin{itemize}

\item For Section~\ref{sec:param}, our main result is that for any fixed $r \geq 3$, the problem \probrFixSwap is $\W[1]$-hard parameterized by the number of swaps~$k$. Moreover, the same is true for \probrFixSwapPromise. This should be contrasted with the positive FPT result of Junosza-Szaniawski~\emph{et al.}~\cite{Junosza2015} for \probrFix. In addition, we observe both problems \probrFix and \probrFixSwap become $\W[1]$-hard for treewidth when the number of colors $r$ is part of the input. The constructions exhibit gadget ideas we use for the sections to follow.

\item In Section~\ref{sec:kernel}, we prove that under plausible complexity assumptions, \probrFix has no polynomial kernel parameterized by the number of recolorings~$k$, for every $r \geq 3$. We stress that while mentioned as an open problem in~\cite{Junosza2015}, the question was subsequently answered by Garnero, Junosza-Szaniawski, Liedloff, Montealegre, and Rz{\k{a}}{\.z}ewski in a full version~\cite{fix-journal} of~\cite{Junosza2015}. Our result was obtained independently of their work, and uses slightly different ideas.

\item Finally, in Section~\ref{sec:cv}, we consider the complexity of the promise variants of the problems (see Section~\ref{sec:prelims} for precise definitions). We show that for $r=3$, both \probrFixSwapPromise and \probrFixPromise are $\NP$-hard for planar graphs. Moreover, the problems cannot be solved in $2^{o(\sqrt{n})}$ time unless the Exponential Time Hypothesis fails. On the positive side, using known results, we derive an algorithm for the problem working in $2^{O(\sqrt{n})}$ time.
\end{itemize}

\section{Preliminaries}
\label{sec:prelims}

All graphs in this paper are simple and undirected. For graph-theoretic notion not defined here, we refer the reader to~\cite{Diestel2010}. We write $[n]$ to denote the set $\{1,2,\ldots,n\}$.

\subsection{Promises and problem statements}
A \emph{promise problem} is a generalization of a decision problem, where the input is guaranteed to belong to a restricted subset among all possible inputs~\cite{Goldreich2008}.

\begin{definition}[Promise problem] A promise problem is a pair of disjoint sets of strings $(S_Y, S_N)$, and their union $S_Y \cup S_N$ is called the \emph{promise set}. An algorithm $A$ decides a promise problem if for every $x \in S_Y, A(x) = 1$ and for every $x \in S_N, A(x)=0$; for strings that do not belong to the promise set $x \notin S_Y \cup S_N$ the algorithm $A$ must halt, but can answer arbitrarily.
\end{definition}
A promise problem is in $\PromiseNP$, the promise extension of $\NP$, if there exists a polynomial $p$ and a polynomial-time verifier $V$ such that for every $x \in S_Y$ there exists $y$ of length at most $p(|x|)$ such that $V(x,y)=1$ and for every $x \in S_N$ and every $y$ it holds that $V(x,y)=0$. For a more comprehensive treatment, we refer the reader to~\cite{Goldreich2008}.

We are then ready to formally define the problems studied in this work. For the promise variants, a coloring $c'$ is said to be a \emph{permutation} of a proper vertex-coloring $c$ if $c'$ can be obtained from $c$ by a finite number of swaps. In other words, the sizes of the color classes of $c'$ match those of an optimal proper coloring.
\begin{framed}
%\vspace*{-0.20cm}
\noindent \probrFix \\
\textbf{Instance:} A graph $G=(V,E)$, an $r$-coloring $c : V \to [r]$, and a positive integer $k$. \\ 
\textbf{Question:} Can $c$ be made into a proper $r$-coloring of $G$ using at most $k$ recolorings?
%\vspace*{-0.25cm}
\end{framed}

\begin{framed}
%\vspace*{-0.20cm}
\noindent \probrFixPromise \\
\textbf{Instance:} A graph $G=(V,E)$, an $r$-coloring $c : V \to [r]$, and a positive integer $k$. \\ 
\textbf{Promise:} $\chi(G)=r$, and $c$ is a permutation of an optimal proper vertex-coloring of $G$. \\
\textbf{Question:} Can $c$ be made into a proper $r$-coloring of $G$ using at most $k$ recolorings?
%\vspace*{-0.25cm}
\end{framed}

\noindent Note that the number of recolorings needed is precisely the minimum Hamming distance between the given coloring $c$ and a valid coloring $c'$ (if existing).

Similarly, we also define \probrFixSwap and \probrFixSwapPromise, where instead of at most $k$ recolorings we have a budget of at most $k$ swaps. 
\begin{framed}
%\vspace*{-0.20cm}
\noindent \probrFixSwap \\
\textbf{Instance:} A graph $G=(V,E)$, an $r$-coloring $c : V \to [r]$, and a positive integer $k$. \\ 
\textbf{Question:} Can $c$ be made into a proper $r$-coloring of $G$ using at most $k$ swaps?
%\vspace*{-0.25cm}
\end{framed}

\begin{framed}
%\vspace*{-0.20cm}
\noindent \probrFixSwapPromise \\
\textbf{Instance:} A graph $G=(V,E)$, an $r$-coloring $c : V \to [r]$, and a positive integer $k$. \\ 
\textbf{Promise:} $\chi(G)=r$, and $c$ is a permutation of an optimal proper vertex-coloring of $G$. \\
\textbf{Question:} Can $c$ be made into a proper $r$-coloring of $G$ using at most $k$ swaps?
%\vspace*{-0.25cm}
\end{framed}

%\noindent Note that the number of recolorings needed is precisely the minimum Hamming distance between the given coloring $c$ and a valid coloring $c'$ (if existing). %Similarly, the number of swaps needed is precisely half the minimum Hamming distance of $c$ and $c'$. In the promise variants, such a $c'$ is guaranteed to exist.

\noindent At first glance, the promise conditions might seem to make the two problems \probrFixPromise and \probrFixSwapPromise similar: one could think that two recolorings correspond to a swap because if we recolor a vertex, then by the promise, the color must be reinserted elsewhere. However, it is easy to build graphs in which this does not hold. For example, consider a graph $G$ constructed from a triangle with the vertices $v_1$, $v_2$, $v_3$ colored $c_1$, $c_2$, $c_3$, respectively. Also, connect vertices to $G$ such that $v_1$ has three pendants (neighbours) colored $c_1$ and three pendants colored $c_2$; $v_2$ has three pendants colored $c_2$ and three pendants colored $c_3$; and $v_3$ has three pendants colored $c_3$ and three pendants colored $c_1$. Clearly, three recolorings are enough to get a proper coloring. Indeed, we color $v_1$ with $c_3$, $v_2$ with $c_1$, and $v_3$ with $c_2$, and are done. 
In contrast, in the swap variant, two swaps needed (e.g., swap colors on $v_1$ and $v_3$, and then colors on $v_3$ and $v_2$).

\subsection{FPT-reductions and (kernelization) lower bounds}
In this subsection, we briefly review the necessary basics of parameterized complexity.

\begin{definition}
Let $A,B \subseteq \Sigma^* \times \mathbb{N}$ be parameterized problems. A \emph{parameterized reduction} from $A$ to $B$ is an algorithm such that given an instance $(x,k)$ of $A$, it outputs an instance $(x',k')$ of $B$ such that
\begin{enumerate}
\item $(x,k)$ is a YES-instance of $A$ iff $(x',k')$ is a YES-instance of $B$,
\item $k' \leq g(k)$ for some computable function $g$, and
\item the running time is $f(k) \cdot |x|^{O(1)}$ for some computable function $f$.
\end{enumerate}
\end{definition}
In the \probClique problem, we are given a graph $G$ and an integer~$k$. 
The task is to decide whether $G$ contains a complete subgraph on~$k$ vertices.
The class of problems reducible to \probClique under parameterized reductions is denoted by $\W[1]$. 
We define hardness and completeness analogously to classical complexity, but assume parameterized reductions. 
That is, a problem is said to be \emph{$\W[1]$-hard} if \probClique (and thus each problem in $\W[1]$) can be reduced to it by a parameterized reduction. It is widely believed that $\FPT \neq \W[1]$.

Let us recall the well-known Exponential Time Hypothesis (ETH), which is often the assumption used for excluding the existence of algorithms that are considerably faster than e.g., brute-force.
\begin{conjecture}[Exponential Time Hypothesis~\cite{eth}]
There exists a constant $c > 0$, such that 
there is no algorithm solving \probThreeSAT in time $O^*(2^{cn})$, where $n$ is the number of variables.
\end{conjecture}
Suppose $\varphi$ is an instance of \probThreeSAT with $n$ variables and $m$ clauses.
It holds that if there is a linear reduction from \probThreeSAT to, say, a graph problem $X$, then the problem $X$ cannot be solved in time $2^{o(n'+m')}$, where $n'$ and $m'$ denote the number of vertices and edges, respectively. Similar reasoning can be applied for $\W[1]$-hard problems. For instance, it is known that there is no $f(k) n^{o(k)}$-time algorithm for \probIndSet for any computable function $f$, unless ETH fails. Then the existence of a parameterized reduction with a linear parameter dependence from \probIndSet to a problem $X'$ implies a lower bound for $X'$ under ETH. For more examples and discussion, we refer the reader to~\cite{fptbook,Lokshtanov2011}.

Finally, let us then mention the machinery we use to obtain kernelization lower bounds later on (for more details, see~\cite{BodlaenderJK14,fptbook}). 
%\begin{definition}
%An equivalence relation $\mathcal{R}$ on the set $\Sigma^*$ is called a \emph{polynomial equivalence relation} if the following conditions are satisfied:
%\begin{enumerate}
%\item There exists an algorithm that, given strings $x,y \in \Sigma^*$, resolves whether $x \equiv_{\mathcal{R}} y$ in time polynomial in $|x| + |y|$.
%\item Relation $\mathcal{R}$ restricted to the set $\Sigma^{\leq n}$ has at most $p(n)$ equivalence classes, for some polynomial $p(\cdot)$.
%\end{enumerate}
%\end{definition}

\begin{definition}
An equivalence relation $\mathcal{R}$ on $\Sigma^*$ is a \emph{polynomial equivalence relation} if (1) there is an algorithm that given two strings $x,y \in \Sigma^*$ decides whether $\mathcal{R}(x,y)$ in $(|x|+|y|)^{O(1)}$ time; and (2) for any finite set $S \subseteq \Sigma^*$ the equivalence relation $\mathcal{R}$ partitions the elements of $S$ into at most $(\max_{x \in S} |x|)^{O(1)}$ classes.
\end{definition}

%\begin{definition}
%Let $L \subseteq \Sigma^*$ be a language and let $Q \subseteq \Sigma^* \times \mathbb{N}$ be a parameterized language. We say that $L$ \emph{cross-composes} into $Q$ if there exists a polynomial equivalence relation $\mathcal{R}$ and an algorithm $\mathcal{A}$, called the \emph{cross-composition}, satisfying the following conditions. The algorithm $\mathcal{A}$ takes as input a sequence of strings $x_1,x_2,\ldots,x_t \in \Sigma^*$ that are equivalent with respect to $\mathcal{R}$, runs in time polynomial in $\Sigma_{i=1}^{t} |x_i|$, and outputs one instance $(y,k) \in \Sigma^* \times \mathbb{N}$ such that:
%\begin{enumerate}
%\item $k \leq p(\max_{i=1}^{t} |x_i| + \log t)$ for some polynomial $p(\cdot)$, and
%\item $(y,k) \in Q$ iff there exists at least one index $i$ such that $x_i \in L$.
%\end{enumerate}
%\end{definition}

\begin{definition}[Cross-composition]
Let $L \subseteq \Sigma^*$ and let $Q \subseteq \Sigma^* \times \mathbb{N}$ be a parameterized problem. We say $L$ \emph{cross-composes} into $Q$ if there is a polynomial equivalence relation $\mathcal{R}$ and an algorithm which, given $t$ strings $x_1,x_2,\ldots,x_t$ belonging to the same equivalence class of $\mathcal{R}$, computes an instance $(x^*,k^*) \in \Sigma^* \times \mathbb{N}$ in time polynomial in $\Sigma_{i=1}^{t} |x_i|$ such that (1) $(x^*,k^*) \in Q$ iff $x_i \in L$ for some $i \in [t]$; and (2) $k'$ is bounded polynomially in $\max_{i=1}^{t} |x_i| + \log t$.
\end{definition}

%\begin{definition}
%A \emph{polynomial compression} of a parameterized language $Q \subseteq \Sigma^* \times \mathbb{N}$ into a language $R \subseteq \Sigma^*$ is an algorithm that takes as input as instance $(x,k) \in \Sigma^* \times \mathbb{N}$, works in time polynomial in $|x|+k$, and returns a string $y$ such that:
%\begin{enumerate}
%\item $|y| \leq p(k)$ for some polynomial $p(\cdot)$, and
%\item $y \in R$ iff $(x,k) \in Q$.
%\end{enumerate}
%\end{definition}

\begin{theorem}[\hspace{1sp}\cite{BodlaenderJK14}]
\label{thm:or-comp-containment}
Assume that an $\NP$-hard language $L$ cross-composes into a parameterized language $Q$. Then $Q$ does not admit a polynomial kernel, unless $\NP \subseteq \coNP / \poly$ and the polynomial hierarchy collapses.
\end{theorem}

\section{Parameterized aspects of restoring corrupted colorings}
\label{sec:param}
Junosza-Szaniawski~\emph{et al.}~\cite{Junosza2015} focused on the \probrFix problem, that is, the number of colors in the coloring is fixed to be $r$. Among other results, they showed the problem is FPT parameterized by treewidth. In other words, the \probFix problem (same as \probrFix but $r$ is part of the input) is FPT for the combined parameter $(r+t)$, where $t$ is the treewidth of the input graph. In contrast, we observe the problem is $\W[1]$-hard when the parameter is only treewidth. 

In the \probPrExt problem (\probPrExtShort), we are given a graph $G=(V,E)$, a set $W \subseteq V$ of precolored vertices, and a precoloring $c : W \to [r]$ of the vertices in $W$. 
The goal is to decide whether there is a proper $r$-coloring $c'$ of $G$ extending the coloring $c$ (i.e., $c'(v) = c(v)$ for every $v \in W$). 
When $r$ is fixed, we call the problem \probkPrExt (\probkPrExtShort). 
Let us then proceed with the following observation.

\begin{lemma}%[$\star$]
\label{lem:kprext_to_fix}
There exists a polynomial time algorithm which given an instance $I = (G,W,c,r)$ of \probPrExt constructs an instance $I' = (G',r',c',k')$ of \probFix, such that $I$ is a YES-instance of \probPrExt iff $I'$ is a YES-instance of \probFix. 
\end{lemma}
\begin{proof}
Let $I = (G,W,c,r)$ be an instance of \probPrExt. 
To obtain, in polynomial time, an instance $I' = (G',r',c',k')$ of \probFix, we proceed as follows. 
First, let $G' = G$, $r' = r$, and set the number of recolorings $k' = |V \setminus W|$.
Then to each precolored vertex $w \in W$, we attach $(r'-1) \cdot (k'+1)$ pendant vertices, called $P_w$.
Build $c'$ from $c$ as follows.
We color vertices in $P_w$ such that there are precisely $k'+1$ vertices colored in every color $c \in ([r'] \setminus \{ c(w) \})$.
We retain the colors on the vertices in $W$, and color each uncolored vertex with color $1$. 
Observe that $k'$ recolorings will not suffice to change the color of $w \in W$ as it has $k'+1$ pendants colored in each color distinct from $c(w)$.
Thus, it is easy to see $I = (G,W,c,r)$ is a YES-instance of \probPrExt iff $I' = (G',r',c',k')$ is a YES-instance of \probFix.
\end{proof}

To make the result hold for \probFixSwap, we add $r \cdot k'$ isolated vertices and color them so that we provide a choice of one of the $r$ colors for each of the $k'$ non-precolored vertices.
Finally, it is well-known the addition of vertices of degree at most~1 does not increase the treewidth of a graph. 
As \probPrExt is $\W[1]$-hard for treewidth~\cite{Fellows2011}, we obtain the following.
\begin{corollary}
Both problems \probFix and \probFixSwap are $\W[1]$-hard parameterized by treewidth.
\end{corollary}
Because \probPrExt is $\NP$-complete when restricted to distance-hereditary graphs~\cite{Bonomo2008} (and thus for e.g., chordal graphs), we immediately observe the following.
\begin{corollary}
Both problems \probFix and \probFixSwap are $\NP$-complete when restricted to the class of distance-hereditary graphs.
\end{corollary}
The above also implies hardness for bounded cliquewidth graphs.

Junosza-Szaniawski~\emph{et al.}~\cite{Junosza2015} proved that for every fixed $r$, the problem \probrFix is FPT parameterized by the number of recolorings. However, when the basic operation is a swap instead of a recoloring, the problem becomes hard. This is established by the following lemma. For the result, we give a parameterized reduction from the well-known \probIndSet problem. In this problem are given a graph $G=(V,E)$, and an integer $k$.
The goal is to decide whether $G$ contains a set of $k$ pairwise non-adjacent vertices. 
The problem is well-known to be $\W[1]$-hard parameterized by~$k$.

\begin{lemma}
\label{lem:indset_to_rswap}
There exists a polynomial time algorithm which given an instance $I = (G,k)$ of \probIndSet constructs an instance $I' = (G',r,c,k')$ of \probrFixSwap for $r = 3$ such that $I$ is a YES-instance of \probIndSet iff $I'$ is a YES-instance of \probrFixSwap.
\end{lemma}
\begin{proof}
Let $V(G) = \{ u_1,u_2,\ldots,u_n \}$.
To construct the graph $G'$, begin with $V(G') = V(G)$ and $k' = 2k$.
Add $k$ disjoint triangles $\{a_1,b_1,c_1\}, \ldots \{ a_k,b_k,c_k\}$, and color them such that $c(a_j) = 3$ and $c(b_j) = c(c_j) = 2$, for $j \in [k]$.
For each $i \in [n]$, add a disjoint triangle $C_i = \{ u^i_a, u^i_b, u^i_c \}$.
In each, color $c(u^i_a) = 1$, $c(u^i_b) = 2$, and $c(u^i_c) = 3$. 
Attach $2(k+1)$ pendant vertices to $u^i_a$ and color $k+1$ of them with color $2$, and $k+1$ of them with color $3$. 
For each $i$, add the edge $u_iu^i_b$.
For every $i,j \in [n]$, if $j > i$ and $u_iu_j \in E(G)$, add the edge $u_ju^i_c$.
Finally, for each $i \in [n]$, add $k+1$ disjoint triangles $\{ t^i_{a,1}, t^i_{b,1}, t^i_{c,1}\}, \ldots, \{ t^i_{a,k+1}, t^i_{b,k+1}, t^i_{c,k+1}\}$. 
These $k+1$ disjoint triangles are colored such that $c(t^i_{a,j}) = 3$, $c(t^i_{b,j}) = 2$, and $c(t^c_{a,j}) = 1$, where $j \in [k+1]$.
Add the edge $u_it^i_{a,j}$, for $i \in [n]$ and $j \in [k+1]$.
This completes the construction of $G'$.
An example is shown in Figure~\ref{fig:indset-to-swap}.
Let us then prove $I = (G,k)$ is a YES-instance of \probIndSet iff $I' = (G',r,c,k')$ is a YES-instance of \probrFixSwap.

\begin{figure}[t]
\centering
\includegraphics[angle=270,width=\textwidth]{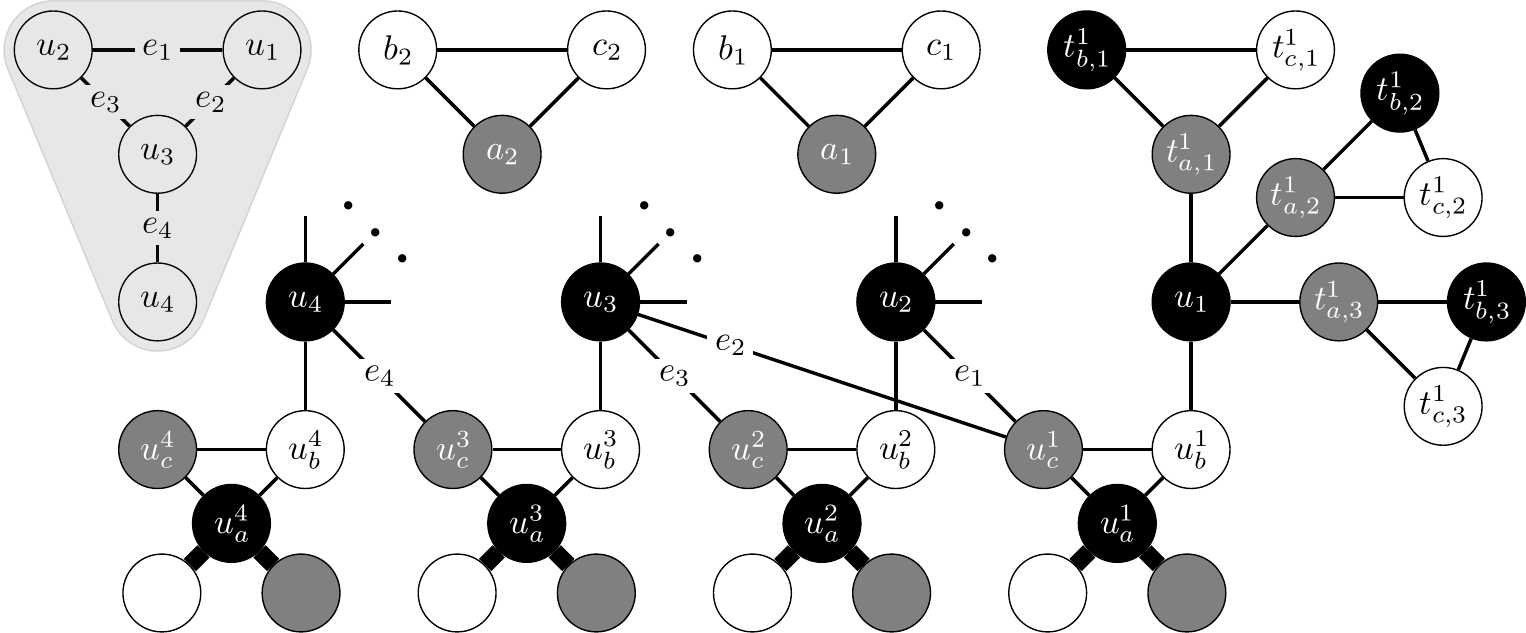}
\caption{An instance $(G,k)$ of \probIndSet with $k=2$ (highlighted in gray) transformed to an instance of \probrFixSwap. Thick edges at the bottom correspond to $k+1$ pendant vertices. A gadget has been expanded for $u_1$: the three dots hide a similar gadget for each of $u_2$, $u_3$, and $u_4$. The colors are ``\protect\tikz[baseline=0.25ex] \protect\fill[draw,fill=black] (1ex,1ex) circle (1ex); $= 1$'', ``\protect\tikz[baseline=0.25ex] \protect\fill[draw,fill=white] (1ex,1ex) circle (1ex); $= 2$'', and ``\protect\tikz[baseline=0.25ex] \protect\fill[draw,fill=gray] (1ex,1ex) circle (1ex); $= 3$''.}
\label{fig:indset-to-swap}
\vspace{-0.25cm}
\end{figure}

Let $S = \{ s_1,\ldots,s_k\} \subseteq V(G)$ be an independent set.
By construction, there are exactly $k$ conflicts in $G'$ between $b_j$ and $c_j$, for $j \in [k]$.
Swap $c_j$ with $u_j$, spending a total of $k$ swaps.
By doing so, we fixed $k$ conflicts but also introduced $k$ new conflicts.
In particular, the new conflicts are between $u_i$ and $u^i_b$, as they are both colored $2$.
We swap $u^i_b$ with $u^i_c$ (colored $3$), and claim $G'$ is properly colored.
By construction, $u^i_c$ is adjacent to $u \in V(G)$ only if $u$ and $u_i$ are adjacent in $V(G)$.
As $S$ is an independent set, each $u \in V(G)$ adjacent to $u^i_c$ is colored $1$.
Thus, $c$ is a proper coloring for $G'$, and we are done. 

For the other direction, suppose $k' = 2k$ swaps suffice to obtain a proper vertex coloring from $c$.
Again, each of the $k$ conflicts between $b_j$ and $c_j$, for $j \in [k]$, must be fixed.
To resolve the conflicts, we must swap either $b_j$ or $c_j$ with a vertex colored $1$.
Without loss, suppose we choose $c_j$ over $b_j$.
Observe that for every $i \in [n]$, we cannot swap $u^i_a$ with any vertex as it has $k+1$ pendant vertices colored $2$, and $k+1$ pendant vertices colored $3$.
Thus, there are two possibilities: either we swap $c_j$ with $u_i$, or with $t^i_{a,j}$, for some $i \in [n]$ and $j \in [k+1]$.
If the swap occurs between $c_j$ and $t^i_{a,j}$, we must then swap $t^i_{b,j}$ with $u_i$ for some $i \in [n]$.
This introduces a conflict between the chosen $u_i$ and its adjacent vertex $u^i_b$.
After we fix this conflict, we have used a total of 3 swaps, totaling $3k > k'$ swaps if we followed this strategy for each $c_j$.
Thus, as $k' = 2k$ swaps suffice, and we are looking for an independent set of size exactly~$k$, we must swap each of the vertices $c_j$ with some $u_i$.
Clearly, two vertices $u_i, u_\ell$, for $i,\ell \in [n]$, adjacent in $G$ cannot be used to fix the conflicts, for otherwise we would have a conflict between the vertices $u^i_b$ and $u^\ell_b$ (both colored $2$).
We conclude that the vertices $u_i$ a vertex $c_j$ is swapped with form an independent set.
\end{proof}
By adding a properly colored $r$-clique to the constructed graph, we can extend the lemma to cover every fixed value of $r$. Thus, we have the following.
\begin{theorem}
\label{thm:w1-rswap}
For every $r \geq 3$, the problem \probrFixSwap is $\W[1]$-hard parameterized by the number of swaps~$k$. Furthermore, there is no $f(k) n^{o(k)}$-time algorithm for the problem unless ETH fails, where $f$ is a computable function.
\end{theorem}
In addition, it is straightforward to extend the construction of Lemma~\ref{lem:indset_to_rswap} to show the promise variant, namely \probrFixSwapPromise, is $\W[1]$-hard parameterized by the number of swaps~$k$.
\begin{corollary}
\label{cor:rfixswappromise-w1hard}
For every $r \geq 3$, the problem \probrFixSwapPromise is $\W[1]$-hard parameterized by the number of swaps~$k$. Furthermore, there is no $f(k) n^{o(k)}$-time algorithm for the problem unless ETH fails, where $f$ is a computable function.
\end{corollary}
\begin{proof}
Assume the construction of an instance $I=(G',r,c,k')$ of \probrFixSwap of Lemma~\ref{lem:indset_to_rswap}. To prove the claim, it suffices to augment the construction to ensure the promise holds, i.e., that with a finite number of swaps we have that $G'$ is properly $r$-colored and $\chi(G') = r$.

A \emph{double star} $S'_n$ is the complete bipartite graph $K_{1,n}$ with each edge subdivided. Formally, $S'_n = (\{s\} \cup \{q_1,...,q_n\} \cup \{q'_1,...,q'_n\},  \{ (s,q_i) ,(q_i,q'_i) \mid i \in [n] \})$. To enforce the promise, add $n$ disjoint double stars $S'_{k+1}$ to $G'$. Each double star is colored such that the central vertex $s$ receives color 1, vertices $q_i$ adjacent to the central vertex color 2, and the remaining vertices $q'_i$ color 3. Observe that after swapping the central vertex $s$ (colored~1) of a $S_{k+1}$ with (say) $b_1$ (colored~2), we require $k'+1$ more swaps to fix the conflicts residing at that particular $S'_{k+1}$.
Nevertheless, given enough swaps, $k' \cdot (k'+1)$ to be precise, we can properly $r$-color $G'$. Finally, to guarantee $\chi(G') = r$, it suffices to add a disjoint properly $r$-colored clique.
\end{proof}

\section{No polynomial kernel for \probrFix}
\label{sec:kernel}
Junosza-Szaniawski~\emph{et al.}~\cite{Junosza2015} showed that for any fixed $r$, the problem \probrFix is FPT parameterized by the number of recolorings $k$. In particular, their result implies a kernel of exponential size for the problem. Thus, they asked whether or not there is a kernel of polynomial size. The question was answered in the negative in a full version of~\cite{Junosza2015} by Garnero~\emph{et al.}~\cite{fix-journal}.
Independently of their work, in what is to follow, we give an alternative proof of the theorem.

\begin{lemma}
For $r = 3$, the problem \probrFix parameterized by the number of recolorings $k$ does not admit a polynomial kernel unless $\NP \subseteq \coNP / \poly$.
\end{lemma}
\begin{proof}
We show that \probThreeSAT cross-composes into \probrFix parameterized by the number of recolorings $k$. By choosing an appropriate polynomial equivalence relation $\mathcal{R}$, we can assume we are given a sequence $\varphi_1,\varphi_2,\ldots,\varphi_t$ of \probThreeSAT instances with an equal number of variables, denoted by $n$, and an equal number of clauses, denoted by $m$. 

Let us then proceed with an cross-composition algorithm that composes $t$ input instances $\varphi_1$,$\varphi_2$,$\ldots,\varphi_t$ which are equivalent under $\mathcal{R}$ into a single instance of \probrFix parameterized by the number of recolorings. 
Specifically, we construct an instance $(G,k)$ of \probrFix, where $G$ is a vertex-colored graph, and $k$ the number of recolorings. 
Our plan is to convert each \probThreeSAT instance $\varphi_1,\varphi_2,\ldots,\varphi_t$ to an instance $\varphi'_1,\varphi'_2,\ldots,\varphi'_t$ of \probFourSAT. 
For each resulting instance of \probFourSAT, we apply the standard reduction from \probFourSAT to \probThreeCol (see e.g.,~\cite{Garey1976}). 
Finally, the resulting graphs are connected by a \emph{spread gadget}, which acts as an instance selector. 
Let us first describe the gadgets, and then the construction of the whole graph $G$.
At the same time, we describe a 3-coloring $c : V(G) \to [3]$.
We set $k = 2 \log_2(t) + 2n + 9m$.
Our construction depends crucially on $k$, and its choice will become apparent later on.

\heading{3-SAT to 4-SAT} 
For each \probThreeSAT formula $\varphi_h$, where $h \in [t]$, introduce a new variable $u_h$, and add it to each clause of $\varphi_h$. 
We call the resulting \probFourSAT formula $\varphi_h'$.
Observe that by setting $u_h$ to true we satisfy $\varphi_h'$.

\heading{The variable vertices} 
Let the $n$ variables of $\varphi_h$, where $h \in [t]$, be $x_{1,h},x_{2,h},\ldots,x_{n,h}$.
We introduce $n$ disjoint 2-cliques labelled $\{x_{1,h},\neg x_{1,h}\}, \ldots, \{ x_{n,h}, \neg x_{n,h} \}$. 
Set $c(x_{i,h}) = 2$ and $c(\neg x_{i,h}) = 1$, for $i \in [n]$ (i.e., initially we set each variable true).
Add an isolated vertex $u_h$, and let $c(u_h) = 1$.
We refer to each of the $2n+1$ vertices as \emph{variable vertices}.

% Colors:
% R = c_1 = 0
% G = c_2 = 1
% B = c_3 = 2

\heading{The clause gadget}
Denote by $C_{h,j}$ the $j$th clause of $\varphi_h'$.
For each clause $C_{h,j}$, where $h \in [t]$ and $j \in [m]$, construct the following clause gadget $H_{h,j}$ (see Figure~\ref{fig:kernel_gadgets}~(b)). 
Take three disjoint triangles $\{a_j,b_j,y_{1,j}\}$, $\{c_j,d_j,y_{2,j}\}$, $\{y_{3,j},y_{4,j},r_j\}$, and add the edges $ y_{1,j}y_{4,j}$, $y_{2,j}y_{3,j}$, and $y_{5,j}r_j$. 
We add $k+1$ pendant vertices adjacent to $r_j$ and color them with color $1$. 
This guarantees $k$ recolorings cannot give $r_j$ color $1$.
Vertices $a_j$, $b_j$, $c_j$ and $d_j$ correspond to the 4 literals each clause has.
Thus, we connect them to the corresponding variable vertices.
That is, when $u \in \{ a_j,b_j,c_j,d_j \}$ corresponds to the variable $x_{i,h}$, $i \in [n]$, we add the edge $ux_{i,h}$ (and similarly when its negated).
The following properties hold for a clause gadget~$H_{h,j}$.
\begin{enumerate}
\item [$(P_1)$] If all four variable vertices of $H_{h,j}$ have color~1, then $r_j$ must have color~1 (costing $k+1$ recolorings to properly color the gadget).

\item [$(P_2)$] The gadget $H_{h,j}$ can be properly 3-colored if one of the attached variable vertices (including $u_h$) have color~2.

%\item [$(P_1)$] In any valid proper 3-coloring of $H_{h,j}$, it holds that at least one of $a_j$, $b_j$, $c_j$, and $d_j$ has color $1$.
%\item [$(P_2)$] If $c(u_h) = 2$, then the partial 3-coloring given to $H_{h,j}$ can be extended to a valid proper 3-coloring for it.
\end{enumerate}

\heading{The spread gadget} 
The spread gadget is constructed by starting from a complete binary tree on $t$ leaves $\ell_1,\ell_2,\ldots,\ell_t$ with the root $r$.
We replace each internal vertex with a triangle, and attach $k+1$ pendant vertices to $r$.
Thus, the distance from $r$ to any leaf is $2 \log_2(t)$. 
We color root $r$, its pendant vertices, and each leaf $\ell_1,\ell_2,\ldots,\ell_t$ with color $1$. 
In a triangle, the top vertex receives color~$1$, the right vertex color~$2$, and the left vertex color~$3$ (see Figure~\ref{fig:kernel_gadgets}~(b)). 
This finishes the construction of the spread gadget.

\begin{figure}
\begin{minipage}[b]{.5\linewidth}
\centering
\includegraphics[width=0.9\textwidth,keepaspectratio]{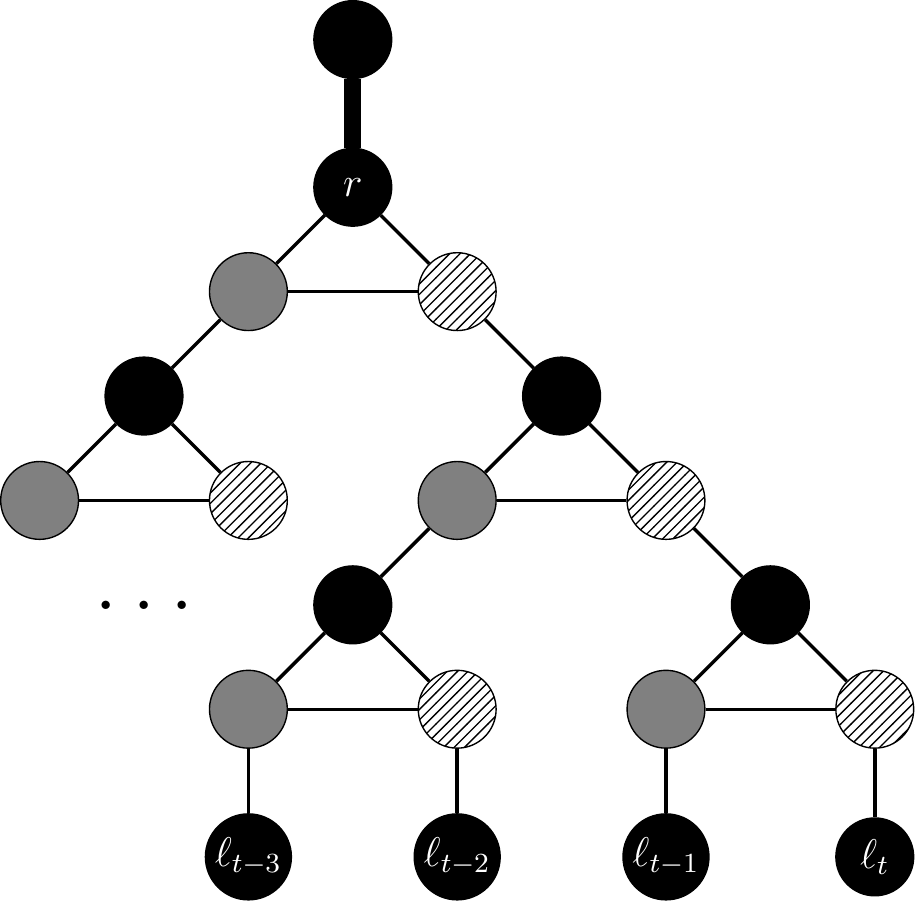}
\end{minipage}%
\begin{minipage}[b]{.5\linewidth}
\centering
\includegraphics[width=0.9\textwidth,keepaspectratio]{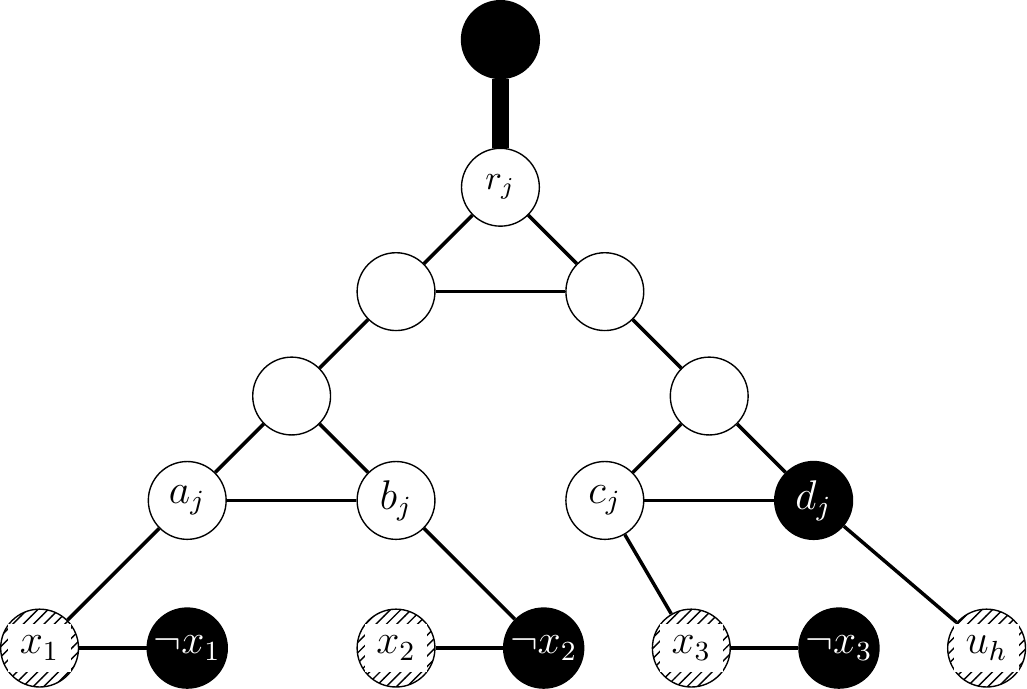}
\end{minipage}
\caption{\textbf{(a)} In the spread gadget, fixing the conflict at the root $r$ corresponds to choosing an instance $\ell_h$, for $h \in [t]$. \textbf{(b)} The clause gadget with its initial vertex-coloring, where the white vertices are colored such that $(P_1)$ and $(P_2)$ are respected. In both figures, thick edges correspond to $k+1$ pendant vertices. The colors are ``\protect\tikz[baseline=0.25ex] \protect\fill[draw,fill=black] (1ex,1ex) circle (1ex); $= 1$'', ``\protect\tikz[baseline=0.25ex] \protect\fill[draw,pattern=north east lines] (1ex,1ex) circle (1ex); $= 2$'', and ``\protect\tikz[baseline=0.25ex] \protect\fill[draw,fill=gray] (1ex,1ex) circle (1ex); $= 3$''.}
\label{fig:kernel_gadgets}
\vspace*{-0.2cm}
\end{figure}

To obtain $G$, we connect the described gadgets together as follows. 
For each $\varphi_h'$, where $h \in [t]$, we add a vertex $w_h$.
Make $w_h$ adjacent to each variable vertex, and set $c(w_h) = 3$.
We give $w_h$ altogether $2(k+1)$ pendant vertices, and color them so that $k+1$ of them have color $1$, and $k+1$ of them have color $2$.
This enforces $c(w_h) = 3$ if only $k$ recolorings are available.
In the spread gadget, each leaf $\ell_h$ is made adjacent to both $u_h$ and $w_h$.

This completes the construction of the graph $G$. 
Recall $k = 2 \log_2(t) + 2n + 9m$, and output the instance $(G,k)$ of \probrFix.
Let us then prove that $(G,k)$ is a YES-instance of \probrFix iff one of $\varphi_h$ is satisfiable for $h \in [t]$.

\heading{Correctness}
Suppose $(G,k)$ is a YES-instance of \probrFix.
By construction, the root $r$ and its $k+1$ pendant vertices are colored with color $1$ in the spread gadget.
As we only have a budget of $k$ recolorings, we must recolor $r$.
By doing so, we introduce a conflict into the triangle containing $r$.
When this conflict is fixed, we move it to one of the two succeeding triangles.
Further continuing to fix the conflict, we propagate it down to one of the leaves $\ell_s$, for some $s \in [t]$.
Intuitively, the propagation to $\ell_s$ means we have chosen to solve the instance $\varphi_s$.
By construction, $\ell_s$ forms a triangle with $u_h$ (colored $2$) and $w_h$ (colored $3$).
As $w_h$ has $k+1$ pendants colored with $1$ and $k+1$ pendants colored with $2$, we must move the conflict to $u_h$.
By moving the conflict from $r$ to $u_h$, we used precisely $2 \log_2(t)$ swaps.
Now, $u_h$ has color $1$, as do all the vertices in $D = \{ d_j \in V(C_{s,j}) \mid j \in [m] \}$.
As $c(u_h) = 1$, we have set the truth value of $u_h$ to false. 
Thus, the truth value of $\varphi_s$ is not affected by the truth value of $u_h$.
As $|D| = m$, each vertex in $D$ can be recolored in $m$ recolorings.
Moreover, $9m$ recolorings suffice to swap the color of each vertex in the three vertex-disjoint triangles a clause gadget has.
By construction, the initial vertex-coloring corresponds to a truth assignment $\tau = \{ z_1,z_2,\ldots,z_n \} \in \{ 1 \}^n$ setting each variable to true.
Clearly, $2n$ recolorings suffice to reverse $\tau$, i.e., change $\tau$ to $\tau'$ such that the Hamming distance of $\tau$ and $\tau'$ is $n$.
Therefore, by $(P_2)$, if $(G,k)$ is a YES-instance of \probrFix, then $\varphi_s$ is satisfiable.

For the other direction, suppose $\varphi_s$ is satisfiable for some $s \in [t]$.
Again, the initial vertex-coloring corresponds to a truth assignment $\tau = \{ z_1,z_2,\ldots,z_n \} \in \{ 1 \}^n$ setting each variable to true.
Using at most $2n+9m$ recolorings, we turn the initial vertex-coloring to a vertex-coloring corresponding to $\tau' = \{ z'_1,z'_2,\ldots,z'_n \} \in \{ 0,1 \}^n$ such that $\varphi_s$ is satisfied under $\tau'$.
Indeed, observe that as $\varphi_s$ is satisfiable, $(P_1)$ is not violated.
Moreover, observe that $\tau'$ satisfies $\varphi_s$ regardless of the truth value of $u_s$.
Thus, we can freely let $c(u_s) = 1$.
But now we introduce a conflict between $\ell_s$ (colored $1$) and its unique neighbor in the spread gadget.
However, using precisely $2 \log_2(t)$ recolorings, we propagate this conflict, and in particular the color $1$, up to the root $r$.
At most $k = 2 \log_2(t) + 2n + 9m$ recolorings have been used, and no conflicts remain in $G$.
Thus, if at least one of $\varphi_1,\ldots,\varphi_t$ is a YES-instance of \probThreeSAT, then $(G,k)$ is a YES-instance of \probrFix. 
This concludes the proof.
\end{proof}
In order to extend the above result to hold for every $r \geq 4$, we attach $(r - 3) \cdot (k+1)$ pendant vertices to each vertex of the construction. These pendant vertices are colored in the obvious way such that each ``original'' vertex must receive exactly one of the colors 1, 2, or 3.
\begin{theorem}[\hspace{1sp}\cite{fix-journal}]
For every $r \geq 3$, the problem \probrFix parameterized by the number of recolorings $k$ does not admit a polynomial kernel unless $\NP \subseteq \coNP / \poly$.
\end{theorem}
Note that in the light of Theorem~\ref{thm:w1-rswap}, the existence of a kernel of any size depending only on~$k$ and~$r$ for \probrFixSwap is highly unlikely.

\section{Chromatic villainy: \probrFixSwapPromise is hard}
\label{sec:cv}
As the main result of this section, we prove that \probThreeCV is $\NP$-hard when restricted to the class of planar graphs. In other words, even with the additional information that some proper vertex-coloring is always obtainable after a finite number of swaps (and no other proper vertex-coloring with less than 3 colors exists), the problem remains hard.

Our reduction will be from the \probkPrExtShort problem, shown to be $\NP$-complete for $r=3$ when restricted to bipartite planar graphs by Kratochv{\'\i}l~\cite{Kratochvil1993}. In fact, although not explicitly stated, the following slightly stronger result is obtained from~\cite{Kratochvil1993}.
\begin{theorem}[\hspace{1sp}\cite{Kratochvil1993}]\label{thm:prext_hard}
The \probkPrExtShort problem is $\NP$-complete for $r=3$ when restricted to the class of bipartite planar graphs, and each precolored vertex has degree 1, that is, $\deg(w) = 1$ for every $w \in W$.
\end{theorem}
The reader should be aware that in the following, we use the color set $\{0,1,2\}$ instead of $[3]$. This will make it more convenient to describe the coloring through modular arithmetic. We are then ready to proceed with the main result of the section.
\begin{theorem}
\label{thm:3chrom_vil_planar}
\probThreeCV is $\NP$-hard when restricted to the class of planar graphs. Moreover, the same is true even when every swap must be between adjacent vertices.
\end{theorem}
\begin{proof}
Let $(G=(V,E), W, c)$ be an instance of \probkPrExtShort, where $G$ is an $n$-vertex bipartite planar graph, $W \subseteq V$ a set of precolored vertices, and $c : W \to \{0,1,2\}$ a precoloring of the vertices in $W$. 
By Theorem~\ref{thm:prext_hard}, we may assume without loss that each (precolored) vertex in $W$ has degree~1. 
Our construction crucially depends on this fact. 
Let $r=3$ be a fixed color bound, and let $h = |V \setminus W|$. 
We will construct a graph $H$ along with its vertex-coloring $c_H$ such that the precoloring $c$ can be extended to a valid $r$-coloring of $G$ iff at most $h$ swaps are needed to transform $c_H$ to an optimal proper vertex-coloring of $H$, i.e., $B(c_H) \leq h$. 
To enforce the promise of the problem, it shall hold for $c_H$ that (\emph{i}) it uses precisely $\chi(H)$ colors, and that (\emph{ii}) by using a finite number of swaps $c_H$ can be transformed into a proper coloring of $H$.

\heading{Construction} 
Let $X = V \setminus W$ be the set of uncolored vertices, let $A$ and $B$ be the bipartition of $G$, that is, $V = A \cup B$, and let us name the set of $r$ colors $C = \{ 0,1,2 \}$. 
%In what follows, we will treat the colors as natural numbers, i.e., we can write $c_1 + 1$ to denote $c_2$. 
The graph $H$ and its vertex-coloring $c_H$ are constructed from $G$ and its precoloring $c$ as follows. 

\begin{itemize}
\item We retain the coloring on the vertices of $W$, that is, $c_H(w) = c(w)$, for every $w \in W$. 

\item If $x \in X$ has one or more neighbors colored with color $i$ (observe it cannot have distinctly colored neighbors), we set $c_H(x) = i$. 
If all neighbors of $x$ are uncolored, we set $c_H(x)= 0$ if $x \in A$. 
Otherwise, $x \in B$, so we set $c_H(x)= 1$. 
For each $x$, we also add two vertices $x_1$ and $x_2$ along with the edges $xx_1$ and $xx_2$. 
We color $c_H(x_1) = (i+1) \bmod 3$ and $c_H(x_2) = (i+2) \bmod 3$.

\item For each precolored vertex $w \in W$ we add $2(h+1)$ new vertices $s_{w,1},\ldots,s_{w,h+1}$ and $t_{w,1},\ldots,t_{w,h+1}$. 
These will be made pendant vertices of $w$ by adding the altogether $2(h+1)$ edges $ws_{w,\ell}$ and $wt_{w,\ell}$ where $\ell \in [h+1]$. 
They receive a color as follows, where $c(w)$ denotes the color of $w$:

\begin{itemize}
\item if $w \in A \wedge c(w) \neq 0$, then $c_H(s_{w,\ell}) = 0$  and $c_H(t_{w,\ell}) = f$, where $f \in C \setminus \{ 0,c(w)\}$;

\item if $w \in B \wedge c(w) \neq 1$, then $c_H(s_{w,\ell}) = 1$  and $c_H(t_{w,\ell}) = g$, where $g \in C \setminus \{ 1,c(w)\}$; and

\item in all other cases $c_H(s_{w,\ell}) = (i+1) \bmod 3$, and $c_H(t_{w,\ell}) = (i+2) \bmod 3$.
\end{itemize}

\item For every precolored vertex $w \in A \wedge c(w) \neq 0$, we add $h$ new vertices $s'_{w,1},\ldots,s'_{w,h}$ with the edges $s_{w,j}s'_{w,j}$, and set $c_H(s'_{w,j}) = c(w)$, where $j \in [h]$.

\item For every precolored vertex $w \in B \wedge c(w) \neq 1$, we add $h$ new vertices $s'_{w,1},\ldots,s'_{w,h}$ with the edges $s_{w,j}s'_{w,j}$, and set $c_H(s'_{w,j}) = c(w)$, where $j \in [h]$.

\item Finally, consider an arbitrary precolored vertex $w \in W$, and its set of $h+1$ pendant vertices $t_{w,j}$. We choose an arbitrary vertex among the $t_{w,j}$ vertices, and call it $v$. Then, we add two vertices $r$ and $r'$ along with the edges $vr$, $vr'$, and $rr'$. These vertices are colored such that $c_H(r) = (c_H(v)+1) \bmod 3$ and $c_H(r') = (c_H(v)+2) \bmod 3$. 
\end{itemize}

\begin{figure}
\begin{minipage}[b]{.5\linewidth}
\centering
\includegraphics[width=0.9\textwidth,keepaspectratio]{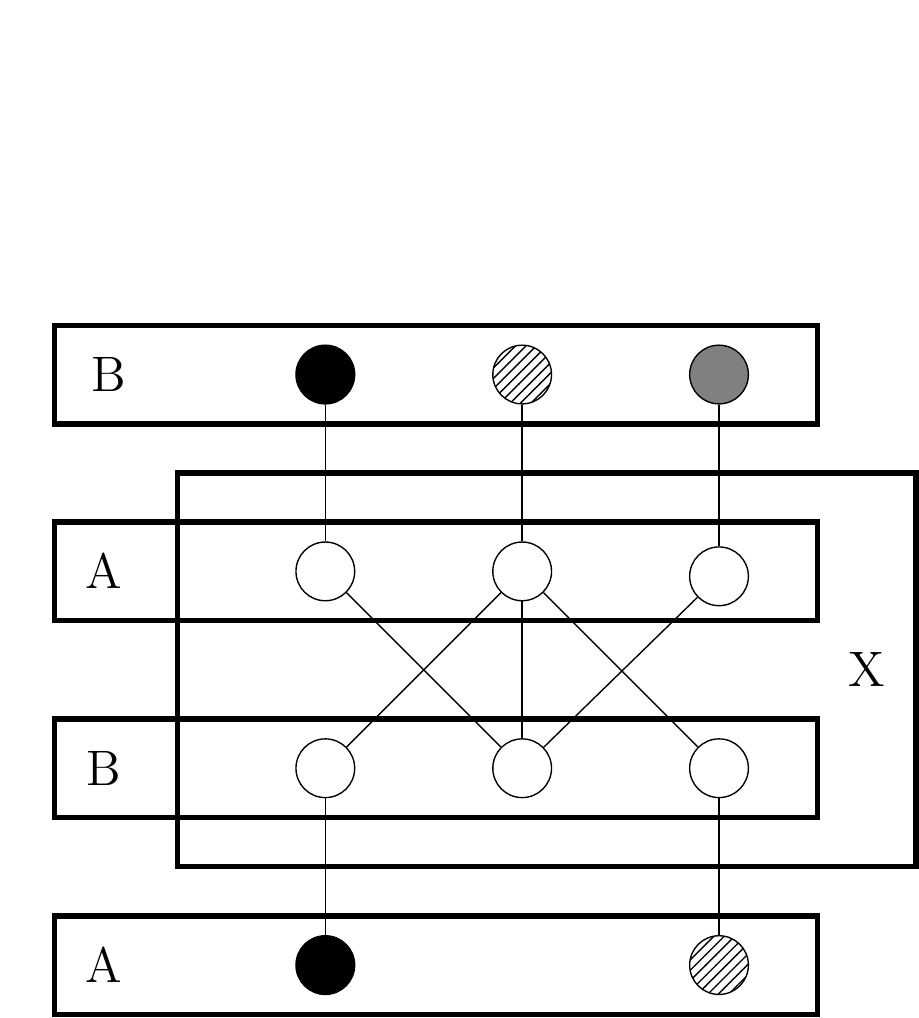}
\end{minipage}%
\begin{minipage}[b]{.5\linewidth}
\centering
\includegraphics[width=0.9\textwidth,keepaspectratio]{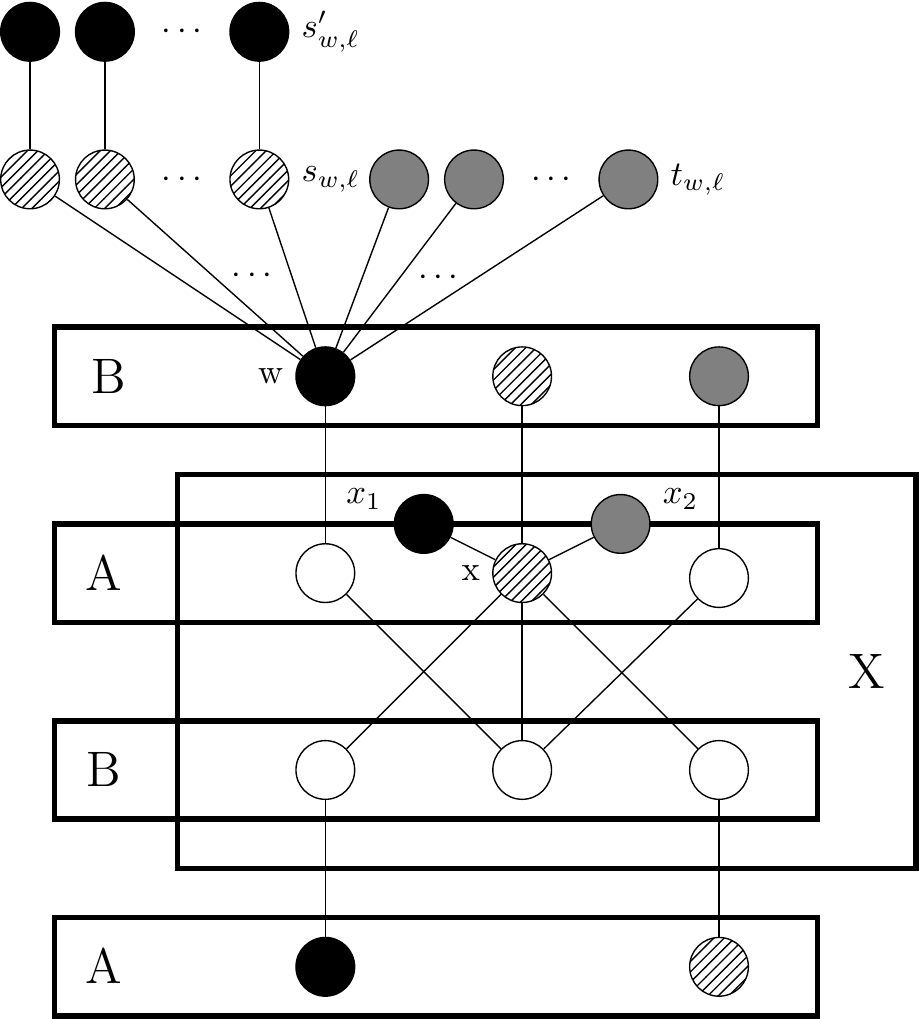}
\end{minipage}
\caption{\textbf{(a)} A partially precolored input graph $G$ of \probkPrExtShort. \textbf{(b)} To reduce clutter, the reduction of Theorem~\ref{thm:3chrom_vil_planar} expanded for only two vertices $x \in X$ and $w \in W$. The colors are ``\protect\tikz[baseline=0.25ex] \protect\fill[draw,fill=black] (1ex,1ex) circle (1ex); $= 0$'', ``\protect\tikz[baseline=0.25ex] \protect\fill[draw,pattern=north east lines] (1ex,1ex) circle (1ex); $= 1$'', and ``\protect\tikz[baseline=0.25ex] \protect\fill[draw,fill=gray] (1ex,1ex) circle (1ex); $= 2$''.}
\label{fig:prext_input}
\vspace{-0.25cm}
\end{figure}

\noindent This finishes the construction of the graph $H$ along with its vertex-coloring $c_H$. 
It is straightforward to verify $G$ is planar, but not bipartite because of the triangle on the vertices $v$, $r$, and $r'$. 
An example is shown in Figure~\ref{fig:prext_input}.
We will then prove $( G=(V,E), W, c )$ is YES-instance of \probPrExtShort iff $B(c_H) \leq h$, that is, if $h$ swaps suffice to transform $c_H$ into a proper coloring of $H$. 

\heading{Correctness} 
Suppose $c$ can be extended to a proper vertex-coloring $c'$ of $G$. 
We will show $h$ swaps suffice to transform $c_H$ into a proper coloring of $H$. 
For each $x \in X$, we perform a swap between $x$ and either $x_1$ or $x_2$. 
Then for every $x$ it holds that either $c_H(x) = c'(x)$ (and there is no need to swap it), or one of the two described swaps can change the color of $x$ to $c'(x)$. 
Now, let $c_H(x)$ be the color of $x$ before the swap, and $c'_H(x) = c'(x)$ its color after the swap. 
Let $u_1,\ldots,u_m$ be the neighbors of $x$. 
If $u_i \in W$ was a precolored vertex, then $c_H(u_i) = c_H(x) \neq c'_H(x)$ by construction; if $u_i \in X$ was an uncolored vertex then the valid coloring $c'$ guarantees that $c'_H(x) \neq c'_H(u_i)$. 
Thus, the claim follows.

For the other direction, suppose $B(c_H) \leq h$. 
Consider a precolored vertex $w \in W$ and let $c_H(w) = i$. 
We claim that for any valid extension $c'$ of $c$, it holds that $c'(w) = c_H(w) = c(w)$.
More precisely, we will show that if the color of $w$ was changed, then it is impossible for $c'$ to be an extension of $c$. 
By construction, the vertex $w$ has $h+1$ neighbors $s_{w,\ell}$ each colored $p$, and $h+1$ neighbors $t_{w,\ell}$ each colored $q$ with $i \neq p \neq q$. 
Thus, if one swap was used to change the color on $w$, then after $h-1$ swaps there would be at least one edge incident to $w$ with its endpoints having the same color. 
So we have that $c'(w) = c_H(w) = c(w)$. 
Moreover, $c'$ is completed to an extension of $c$ by picking the colors $c'_H(x)$ assigned to the uncolored vertices $x \in X$ after the $h$ swaps. 
This completes the proof of correctness for our reduction. 

\heading{Promise}
Let us then show that the promise holds as well. 
That is, we show that $c_H$ can be transformed into a proper vertex-coloring $c''_H$ of $H$ with a finite number of swaps even if the original precoloring of $G$ cannot be extended to a proper coloring (but in this case, more than $h$ swaps are needed). 

First, we show that a finite number of swaps gives us a 2-coloring for $A \cup B$ such that every vertex in $A$ receives color $0$, and every vertex in $B$ color $1$.
Afterwards, we will adjust the remaining pendant vertices $s_{w,\ell}$, $t_{w,\ell}$, and $s'_{w,j}$ so that no color conflict remains. 

If $x \in X \cap A \wedge c_H(x) \neq 1$, then we swap it with one of its neighbors $x_1$ or $x_2$ and get $c''_H(x)=0$. 
Similarly, if $x \in X \cap B \wedge c_H(x)\neq 2$, a swap with either $x_1$ or $x_2$ gives us $c''_H(x) = 1$. 

Let us then consider the precolored vertices. 
If $w \in W \cap A \wedge c_H(w) \neq 0$, we swap $w$ with $s_{w,h+1}$ which is colored $0$. 
This causes a conflict $c''_H(w) = c_H(s_{w,j}) = 0$, which is fixed by swapping $s_{w,j}$ with $s'_{w,j}$ that are colored $c_H(s'_{w,j}) = c_H(w) \neq 0$. 
Similarly, if $w \in W \cap B \wedge c_H(w) \neq 1$, we swap $w$ with $s_{w,h+1}$ which is colored $1$. 
This causes conflicts $c''_H(w) = c_H(s_{w,j}) = 1$, where $j \in [h]$. 
To fix them, we swap $s_{w,j}$ with $s'_{w,j}$ that are colored $c_H(s'_{w,j}) = c_H(w) \neq 1$. 
Thus, we have that all vertices in $A$ are colored $1$, and all vertices in $B$ are colored $1$. 
Moreover, for all $w \in W$ we have $c''_H(w) \neq c''_H(s_{w,\ell}) \wedge c''_H(w) \neq c''_H(t_{w,\ell})$, where $\ell \in [h+1]$. 
Also, $c''_H(s_{w,j}) \neq c''_H(s_{w,j})$, where $j \in [h]$. 
Thus, $c''_H$ is a valid coloring of $H$ and the triangle $v,r,r'$ guarantees that $\chi(H) = 3$. Thus, the claim follows.
\end{proof}

\noindent By removing the promise condition, we can modify our reduction to obtain the following. 
\begin{corollary}
For every $r \geq 3$, the problem \probrFixSwap is $\NP$-complete for bipartite planar graphs.
\end{corollary}

Another corollary follows by a chain of reductions.
First, Lichtenstein~\cite{Lichtenstein1982} gives a reduction from \probThreeSAT to \probPlanarThreeSAT showing \probPlanarThreeSAT cannot be solved in time $2^{o(\sqrt{n+m})}$, unless ETH fails. 
Continuing to compose reductions, Mansfield~\cite{Mansfield1983} gives a linear reduction from \probPlanarThreeSAT to \probPlanarOneInThreeSAT, which is then similarly reduced by Kratochv{\'\i}l~\cite{Kratochvil1993} to \probkPrExtShort (the result of Theorem~\ref{thm:prext_hard}). 
Finally, it can be verified the construction of Theorem~\ref{thm:3chrom_vil_planar} has linear size, giving us the following.
\begin{corollary}
\label{cor:planarcv_subexp}
There is no algorithm which solves \probPlanarThreeCV in $2^{o(\sqrt{n})}$ time unless ETH fails.
\end{corollary}

\noindent However, on a positive side, we claim that for any fixed $r$, the problem \probPlanarrFix (and its promise variant) can be solved in $2^{O(\sqrt{n})}$ time. To see this, we recall that Junosza-Szaniawski~\emph{et al.}~\cite{Junosza2015} showed that for any fixed~$r$, the optimization variant of \probrFix is solvable in $O(nr^{t+2})$ time on graphs of treewidth~$t$. To leverage this result, it is enough to recall the treewidth of a planar graph is $O(\sqrt{n})$. This implies a $2^{O(\sqrt{n})}$-time algorithm for \probPlanarrFix.

Finally, let us mention that by modifying the construction of Theorem~\ref{thm:3chrom_vil_planar} slightly, one can show a similar result for \probrFixPromise, and its non-promise variant as well.
\begin{theorem}
\label{thm:3swappromise_nphard}
\probThreeSwapPromise is $\NP$-hard when restricted to the class of planar graphs.
\end{theorem}
\begin{proof}
Consider the construction in Theorem~\ref{thm:3chrom_vil_planar}.
For each $x \in X$, remove the edges $xx_1$ and $xx_2$, and add the edge $x_1x_2$.

Observe that it still holds that in $h$ recolorings, it is impossible to recolor a precolored vertex $w \in W$.
Thus, the correctness of the reduction holds.
To see that the promise is enforced, note that it is still true that $\chi'(G) = 3 = r$.
For each $x \in X$, there is a corresponding 2-clique, from which a color distinct from $c'(x)$ can be swapped for $x$.
Thus, the promise is also enforced.
We have a valid reduction, and thus conclude the proof.
\end{proof}
By relaxing the requirement on the promise condition, we establish again the following. We also note that using different ideas, the same conclusion was reached in~\cite{fix-journal}.
\begin{corollary}[\hspace{1sp}\cite{fix-journal}]
For every $r \geq 3$, the problem \probrFix is $\NP$-complete for bipartite planar graphs.
\end{corollary}
As also remarked in~\cite{fix-journal}, it is interesting to contrast the above with the results of~\cite{Junosza2015} where it was shown that when $r=2$, the problem \probrFix is solvable in polynomial time. In other words, if we have a bipartite graph that is \emph{not} colored optimally (i.e. more than two colors are used), fixing the coloring is hard.

\section{Conclusions}
We further investigated the complexity of restoring corrupted colorings, especially from a parameterized perspective. Interestingly, we showed that \probrFixSwap is $\W[1]$-hard parameterized by the number of swaps, while \probrFix is known to be FPT parameterized by the number of recolorings. We believe the problems behave similarly for treewidth. Indeed, we conjecture that \probrFixSwap is $\W[1]$-hard parameterized by treewidth, for every $r \geq 3$. One could also consider other natural basic operations, such as swaps between adjacent vertices.

Finally, it might be interesting to perform a similar study for edge-colored graphs. In particular, how does the complexity of edge recoloring compare to vertex recoloring?

\bibliographystyle{abbrv}
\bibliography{bibliography}

\begin{thebibliography}{10}

\bibitem{Bjorklund2009}
A.~Bj{\"o}rklund, T.~Husfeldt, and M.~Koivisto.
\newblock Set partitioning via inclusion-exclusion.
\newblock {\em SIAM Journal on Computing}, 39(2):546--563, 2009.

\bibitem{BodlaenderJK14}
H.~L. Bodlaender, B.~M.~P. Jansen, and S.~Kratsch.
\newblock Kernelization lower bounds by cross-composition.
\newblock {\em {SIAM} Journal of Discrete Mathematics}, 28(1):277--305, 2014.

\bibitem{Bonomo2008}
F.~Bonomo, G.~Dur{\'a}n, and J.~Marenco.
\newblock Exploring the complexity boundary between coloring and list-coloring.
\newblock {\em Annals of Operations Research}, 169(1):3--16, 2008.

\bibitem{Bonsma2014}
P.~S. Bonsma, A.~E. Mouawad, N.~Nishimura, and V.~Raman.
\newblock The complexity of bounded length graph recoloring and {CSP}
  reconfiguration.
\newblock In {\em Proceedings of the 9th International Symposium on
  Parameterized and Exact Computation, {IPEC} 2014, Wroclaw, Poland, September
  10-12}, pages 110--121, 2014.

\bibitem{Clark2006}
S.~A. Clark, J.~E. Holliday, S.~H. Holliday, P.~Johnson, J.~E. Trimm, R.~R.
  Rubalcaba, and M.~Walsh.
\newblock Chromatic villainy in graphs.
\newblock {\em Congressus Numerantium}, 182:171, 2006.

\bibitem{fptbook}
M.~Cygan, F.~V. Fomin, {\L}.~Kowalik, D.~Lokshtanov, D.~Marx, M.~Pilipczuk,
  M.~Pilipczuk, and S.~Saurabh.
\newblock {\em Parameterized Algorithms}.
\newblock Springer, 2015.

\bibitem{Berg2016}
M.~de~Berg, K.~Buchin, B.~M.~P. Jansen, and G.~J. Woeginger.
\newblock Fine-grained complexity analysis of two classic {TSP} variants.
\newblock In {\em Proceedings of the 43rd International Colloquium on Automata,
  Languages, and Programming, {ICALP} 2016, Rome, Italy, July 11-15}, pages
  5:1--5:14, 2016.

\bibitem{Diestel2010}
R.~Diestel.
\newblock {\em Graph Theory}.
\newblock Springer-Verlag Heidelberg, 2010.

\bibitem{Even1984}
S.~Even, A.~L. Selman, and Y.~Yacobi.
\newblock The complexity of promise problems with applications to public-key
  cryptography.
\newblock {\em Information and Control}, 61(2):159--173, 1984.

\bibitem{Fellows2011}
M.~R. Fellows, F.~V. Fomin, D.~Lokshtanov, F.~Rosamond, S.~Saurabh, S.~Szeider,
  and C.~Thomassen.
\newblock On the complexity of some colorful problems parameterized by
  treewidth.
\newblock {\em Information and Computation}, 209(2):143--153, 2011.

\bibitem{Fellows2012}
M.~R. Fellows, F.~V. Fomin, D.~Lokshtanov, F.~Rosamond, S.~Saurabh, and
  Y.~Villanger.
\newblock {Local search: Is brute-force avoidable?}
\newblock {\em Journal of Computer and System Sciences}, 78(3):707--719, 2012.

\bibitem{Garey1976}
M.~R. Garey, D.~S. Johnson, and L.~Stockmeyer.
\newblock {Some simplified NP-complete graph problems}.
\newblock {\em Theoretical Computer Science}, 1(3):237--267, 1976.

\bibitem{fix-journal}
V.~Garnero, K.~Junosza-Szaniawski, M.~Liedloff, P.~Montealegre, and
  P.~Rz{\k{a}}{\.{z}}ewski.
\newblock Fixing improper colorings of graphs.
\newblock {\em CoRR}, abs/1607.06911, 2016.

\bibitem{Goldreich2006}
O.~Goldreich.
\newblock On promise problems: A survey.
\newblock In {\em Theoretical Computer Science}, volume 3895 of {\em Lecture
  Notes in Computer Science}, pages 254--290. Springer Berlin Heidelberg, 2006.

\bibitem{Goldreich2008}
O.~Goldreich.
\newblock {\em {Computational Complexity: A Conceptual Perspective}}.
\newblock Cambridge University Press, 2008.

\bibitem{eth}
R.~Impagliazzo and R.~Paturi.
\newblock {On the Complexity of $k$-SAT}.
\newblock {\em Journal of Computer and System Sciences}, 62(2):367--375, 2001.

\bibitem{Jensen2011}
T.~R. Jensen and B.~Toft.
\newblock {\em Graph coloring problems}, volume~39.
\newblock John Wiley \& Sons, 2011.

\bibitem{Johnson1985}
D.~S. Johnson.
\newblock {The NP-completeness column: an ongoing guide}.
\newblock {\em Journal of Algorithms}, 6(3):434--451, 1985.

\bibitem{Johnson2014}
M.~Johnson, D.~Kratsch, S.~Kratsch, V.~Patel, and D.~Paulusma.
\newblock Finding shortest paths between graph colourings.
\newblock In {\em Proceedings of the 9th International Symposium on
  Parameterized and Exact Computation, {IPEC} 2014, Wroclaw, Poland, September
  10-12}, pages 221--233, 2014.

\bibitem{Junosza2015}
K.~Junosza-Szaniawski, M.~Liedloff, and P.~Rz{\k{a}}{\.z}ewski.
\newblock Fixing improper colorings of graphs.
\newblock In {\em SOFSEM 2015: Theory and Practice of Computer Science}, pages
  266--276. Springer, 2015.

\bibitem{Khuller2003}
S.~Khuller, R.~Bhatia, and R.~Pless.
\newblock On local search and placement of meters in networks.
\newblock {\em SIAM Journal on Computing}, 32(2):470--487, 2003.

\bibitem{Kratochvil1993}
J.~Kratochv{\'\i}l.
\newblock Precoloring extension with fixed color bound.
\newblock {\em Acta Math. Univ. Comen}, 62:139--153, 1993.

\bibitem{Krokhin2012}
A.~Krokhin and D.~Marx.
\newblock On the hardness of losing weight.
\newblock {\em ACM Transactions on Algorithms}, 8(2):1--18, 2012.

\bibitem{Lichtenstein1982}
D.~Lichtenstein.
\newblock Planar formulae and their uses.
\newblock {\em SIAM Journal on Computing}, 11(2):329--343, 1982.

\bibitem{Lokshtanov2011}
D.~Lokshtanov, D.~Marx, and S.~Saurabh.
\newblock {Lower bounds based on the Exponential Time Hypothesis}.
\newblock {\em Bulletin of the EATCS}, (105):41--72, 2011.

\bibitem{Mansfield1983}
A.~Mansfield.
\newblock {Determining the thickness of graphs is NP-hard}.
\newblock In {\em Mathematical Proceedings of the Cambridge Philosophical
  Society}, volume~93, pages 9--23. Cambridge University Press, 1983.

\bibitem{Marx2004}
D.~Marx.
\newblock Graph colouring problems and their applications in scheduling.
\newblock {\em Electrical Engineering}, 48(1-2):11--16, 2004.

\bibitem{Ostergard2004}
P.~R. {\"O}sterg{\aa}rd.
\newblock On a hypercube coloring problem.
\newblock {\em Journal of Combinatorial Theory, Series A}, 108(2):199--204,
  2004.

\bibitem{Szeider2011}
S.~Szeider.
\newblock {The parameterized complexity of $k$-flip local search for SAT and
  MAX SAT}.
\newblock {\em Discrete Optimization}, 8(1):139--145, 2011.

\bibitem{Valiant1985}
L.~G. Valiant and V.~V. Vazirani.
\newblock {NP is as easy as detecting unique solutions}.
\newblock {\em {Theoretical Computer Science}}, 47:85--93, 1986.

\bibitem{Regs2013}
D.~West.
\newblock {Chromatic Villainy of Graphs}.
\newblock
  http://www.math.illinois.edu/{\textasciitilde}dwest/\\regs/chromvil.html,
  2013 (accessed August 3, 2015).

\bibitem{Wrochna2014}
M.~Wrochna.
\newblock Reconfiguration in bounded bandwidth and treedepth.
\newblock {\em arXiv preprint arXiv:1405.0847}, 2014.

\end{thebibliography}

\end{document}